\documentclass[letter,11pt]{article}
\usepackage[margin=1in]{geometry}
\usepackage{hyperref}
\usepackage{setspace}
\usepackage{microtype}
\usepackage{tikz}
\usepackage{xcolor}
\usepackage[nocompress]{cite}
\usepackage{amstext}
\usepackage{epsfig}
\usepackage{psfrag}
\usepackage{comment}
\usepackage{amsmath}
\usepackage{color}
\usepackage{arcs}
\usepackage{algorithmic}
\usepackage{graphics}
\usepackage[ruled,vlined]{algorithm2e}
\usepackage{amssymb,amsmath}
\newtheorem{claim}{Claim}

\usepackage{thmtools}
\usepackage{thm-restate}

\usepackage{graphicx}
\usepackage{caption}
\usepackage{amssymb}

\usepackage{hyperref}
\usepackage[numbers,sort]{natbib}
\usepackage{cleveref}
\newenvironment{proof}{\textbf{Proof:}}{\hfill$\Box$}
\newtheorem{theorem}{Theorem}
\newtheorem{definition}{Definition}
\newtheorem{lemma}[theorem]{Lemma}
\newtheorem{corollary}[theorem]{Corollary}

\def\freq{\mathrm{freq}}
\def\depth{\mathrm{dep}}
\def\rank{\mathrm{rank}}

\def\rank{\mathrm{rank}}
\def\guest{\mathrm{guest}}
\def\Right{\mathrm{right}}
\def\Left{\mathrm{left}}
\def\host{\mathrm{host}}

\def\cost{\mathrm{cost}}
\def\ct{\mathrm{CT}}
\def\acc{\mathrm{acc\text{-}}}
\def\adj{\mathrm{adj\text{-}}}
\def\acost{\mathrm{amortized}}

\def\ALG{\textsc{Alg}}
\def\ON{\textsc{On}}
\def\MRU{\textsc{MRU}}
\def\Random-Push{\textsc{Random-Push}}
\def\OPT{\textsc{Opt}}
\def\card#1{\lvert #1 \rvert}

\def\exp{\mathbb{E}}

\def\XX{{\tt X}}
\def\YY{{\tt Y}}
\def\Prb{\mathbb{P}}
\def\Rnum{\mathbb{R}}

\def\Expct{\mathbb{E}}

\def\freq{\mathrm{freq}}
\def\depth{\mathrm{dep}}
\def\rank{\mathrm{rank}}

\def\rank{\mathrm{rank}}
\def\guest{\mathrm{guest}}
\def\Right{\mathrm{right}}
\def\Left{\mathrm{left}}
\def\host{\mathrm{host}}

\def\cost{\mathrm{cost}}
\def\acc{\mathrm{acc\text{-}}}
\def\adj{\mathrm{adj\text{-}}}
\def\acost{\mathrm{amortized}}

\def\ALG{\textsc{Alg}}
\def\ON{\textsc{On}}
\def\MRU{\textsc{MRU}}
\def\ONRAND{\textsc{Random-Push}}
\def\ONDET{\textsc{Move-Half}}
\def\OPT{\textsc{Opt}}
\def\card#1{\lvert #1 \rvert}

\def\exp{\mathbb{E}}

\def\ct{\mathrm{CT}}
\def\XX{{\tt X}}
\def\YY{{\tt Y}}
\def\Prb{\mathbb{P}}
\def\Rnum{\mathbb{R}}

\def\Expct{\mathbb{E}}

\title{Push-Down Trees: Optimal Self-Adjusting Complete Trees}

\author{Chen Avin$^1$ \quad Kaushik Mondal$^2$ \quad Stefan Schmid$^3$\\
{\small $^1$ Ben Gurion University of the Negev, Israel\quad$^2$ Indian Institute of Technology Ropar, India}\\ {\small \quad $^3$ Faculty of Computer Science,
University of Vienna, Austria}
}

\date{}
\begin{document}

\maketitle


\begin{abstract}
This paper studies a fundamental algorithmic
problem related to the design of demand-aware networks:
networks whose topologies adjust toward the traffic
patterns they serve, in an online manner.
The goal is to strike a tradeoff
between the benefits of such adjustments
(shorter routes) and their costs (reconfigurations).
In particular, we consider the problem of designing a
self-adjusting tree network which serves
single-source, multi-destination communication.
The problem has interesting connections to self-adjusting
datastructures.
We present two constant-competitive
online algorithms for this problem, one randomized and one
deterministic.
Our approach is based on
a natural notion of
\emph{Most Recently Used~(MRU)} tree, maintaining
a \emph{working set}.
We prove that the working set is a cost lower
bound for any online algorithm, and then
present a randomized algorithm \ONRAND~which
\emph{approximates} such an MRU tree
at low cost,
by pushing less recently
used communication partners down the tree, along
a random walk. Our deterministic algorithm \ONDET~
does not directly maintain an MRU tree, but
its cost is still proportional to the cost of an MRU tree, and also
matches the working set
lower bound.
\end{abstract}

\section{Introduction}\label{sec:intro}

While datacenter networks traditionally rely on a
\emph{fixed} topology, recent optical technologies
enable \emph{reconfigurable} topologies which
can adjust to the demand (i.e., traffic pattern) they
serve \emph{in an online manner}, e.g.~\cite{sigact19,ghobadi2016projector,firefly,eclipse}.
Indeed, the physical topology is emerging as the next frontier
in an ongoing effort to render networked systems
more flexible.

In principle, such topological reconfigurations can  be used
to provide shorter routes between frequently communicating
nodes, exploiting structure in traffic patterns~\cite{sigmetrics20complexity,kandula2009nature,ghobadi2016projector},
and hence to improve performance.
However, the design of self-adjusting networks which dynamically
optimize themselves toward the demand introduces
an algorithmic challenge: an online
algorithm needs to be devised which guarantees an efficient
tradeoff between the benefits (i.e., shorter route lengths)
and costs (in terms of reconfigurations) of topological
optimizations.

This paper focuses on the design of a self-adjusting \emph{complete tree
($\ct$)} network:
a network of nodes (e.g., servers or racks) that forms a complete tree,
and we measure the routing cost in terms of the length of the shortest path between two nodes.
Trees are not only a most fundamental topological structure of
their own merit, but also a crucial building block for more general
self-adjusting network designs: Avin et al.~\cite{avin2017demand,infocom19dan}
recently showed that multiple
tree networks (optimized individually for a single source node)
can be combined to build general networks which
provide low degree and low distortion.
The design of a dynamic single-source multi-destination communication tree,
as studied in this paper, is hence a stepping stone.

The focus on trees is further motivated by a relationship of
our problem to problems
arising in
self-adjusting datastructures~\cite{avin2019toward}:
self-adjusting datastructures such as self-adjusting
search trees~\cite{splaytrees}
have the appealing property that they optimize
themselves to the workload, leveraging temporal locality,
but without knowing the future. Ideally,
self-adjusting datastructures store items which will be accessed (frequently)
 \emph{in the future}, in a way that they can be accessed quickly
 (e.g., close to the root, in case of a binary search tree),
while also accounting for reconfiguration costs.
However, in contrast to most datastructures, in a \emph{network},
the search property is not required: the network supports \emph{routing}.
Accordingly
our model can be seen as a novel flavor of such self-adjusting
binary search trees
where lookup is supported by a \emph{map}, enabling
shortest path routing (more details will follow).

We present a formal model for this problem later,
but a few observations are easy to
make. If we restrict ourselves to the special case of a
\emph{line} network (a ``linear tree''),
the problem of optimally arranging the destinations
of a given single communication source is equivalent to the
well-known \emph{dynamic list update} problem: for such self-adjusting
(unordered) lists,
dynamically optimal online algorithms have been known for a long time~\cite{list-update-upperbound}.
In particular,
the simple move-to-front algorithm which immediately promotes the accessed item to the front of the list,
fulfills the \emph{Most-Recently Used (MRU)} property:
the $i^{\mathit{th}}$ furthest away item from the front
of the list is the $i^{\mathit{th}}$ most recently used item. In the list (and hence on
the line),
this property is enough to guarantee optimality.
The MRU property is related to the so called \emph{working set property}: the cost of accessing
item $x$ at time $t$ depends on the number of distinct items accessed since the last access
of $x$ prior to time $t$, including $x$.
Naturally, we wonder whether the \emph{MRU} property is
enough to guarantee optimality also in our case. The answer turns out to be non-trivial.

A first contribution of this paper is the observation that if
we count only \emph{access} cost (ignoring any rearrangement cost, see Definition \ref{def:cost} for details),
the answer is affirmative: the  most-recently used tree is what is called \emph{access optimal}.
Furthermore, we show that the corresponding
access cost is a lower bound for any algorithm which is dynamically optimal.
But securing this property,
i.e., maintaining the most-recently used items close to the root in the tree,
introduces a new challenge: how to achieve this
\emph{at low cost}? In particular,
assuming that \emph{swapping} the locations of items
comes at a \emph{unit cost},
can the property be maintained at cost proportional to the \emph{access} cost?
As we show, \emph{strictly} enforcing the most-recently used property in a tree is
too costly to achieve optimality.
But, as we will show, when turning to an \emph{approximate} most-recently used property, we are
able to show two important properties: \emph{i)} such an approximation is good enough
to guarantee access optimality; and \emph{ii)} it can be maintained in
expectation using a \emph{randomized} algorithm:
less recently used communication partners are pushed down the tree
along a random walk.

While the most-recently used property is \emph{sufficient},
it is  not necessary: we provide a deterministic algorithm which is dynamically optimal
but does not even maintain the MRU property approximately.
However, its cost is still proportional to the cost of an MRU tree (Definition \ref{def:mru}).

Succinctly, we make the following \textbf{\emph{contributions}}.
First we show a working set lower bound for our problem.
We do so by proving that an MRU tree is \emph{access optimal}.
In the following theorem, let $WS(\sigma)$ denote the
working set of $\sigma$
(a formal definition will follow later).

\begin{restatable}{theorem}{wslowerbound}\label{th:lowerbd}
Consider a request sequence $\sigma$. Any algorithm $\ALG$
serving $\sigma$ using a self-adjusting complete tree,
has cost at least $\cost(\ALG(\sigma)) \ge WS(\sigma)/4$,
where $WS(\sigma)$ is the working set of $\sigma$.
\end{restatable}

Our main contribution is a deterministic online algorithm $\ONDET$ which maintains
a constant competitive self-adjusting
Complete Tree ($\ct$) network.
\begin{restatable}{theorem}{dynamicopt}\label{th:dynamicopt}
$\ONDET$ algorithm is dynamically optimal.
\end{restatable}
Interestingly, $\ONDET$ does not require the
MRU property and hence does not need to maintain MRU tree. This implies that
maintaining a working set on $\ct$s is not a necessary condition for dynamic
optimality, although it is a sufficient one.

Furthermore, we present a dynamically optimal, i.e.,
constant competitive (on expectation) randomized algorithm for self-adjusting
$\ct$s called
\textsc{\ONRAND}. \textsc{\ONRAND} relies on maintaining an approximate MRU tree.
\begin{restatable}{theorem}{pushdowntree}\label{thm:pushdowntree}
The \textsc{\ONRAND} algorithm is dynamically optimal on expectation.
\end{restatable}

\begin{figure*}[t]
\ \begin{centering}
  \begin{tabular}[t]{ccc}
   \includegraphics[width=.32\textwidth]{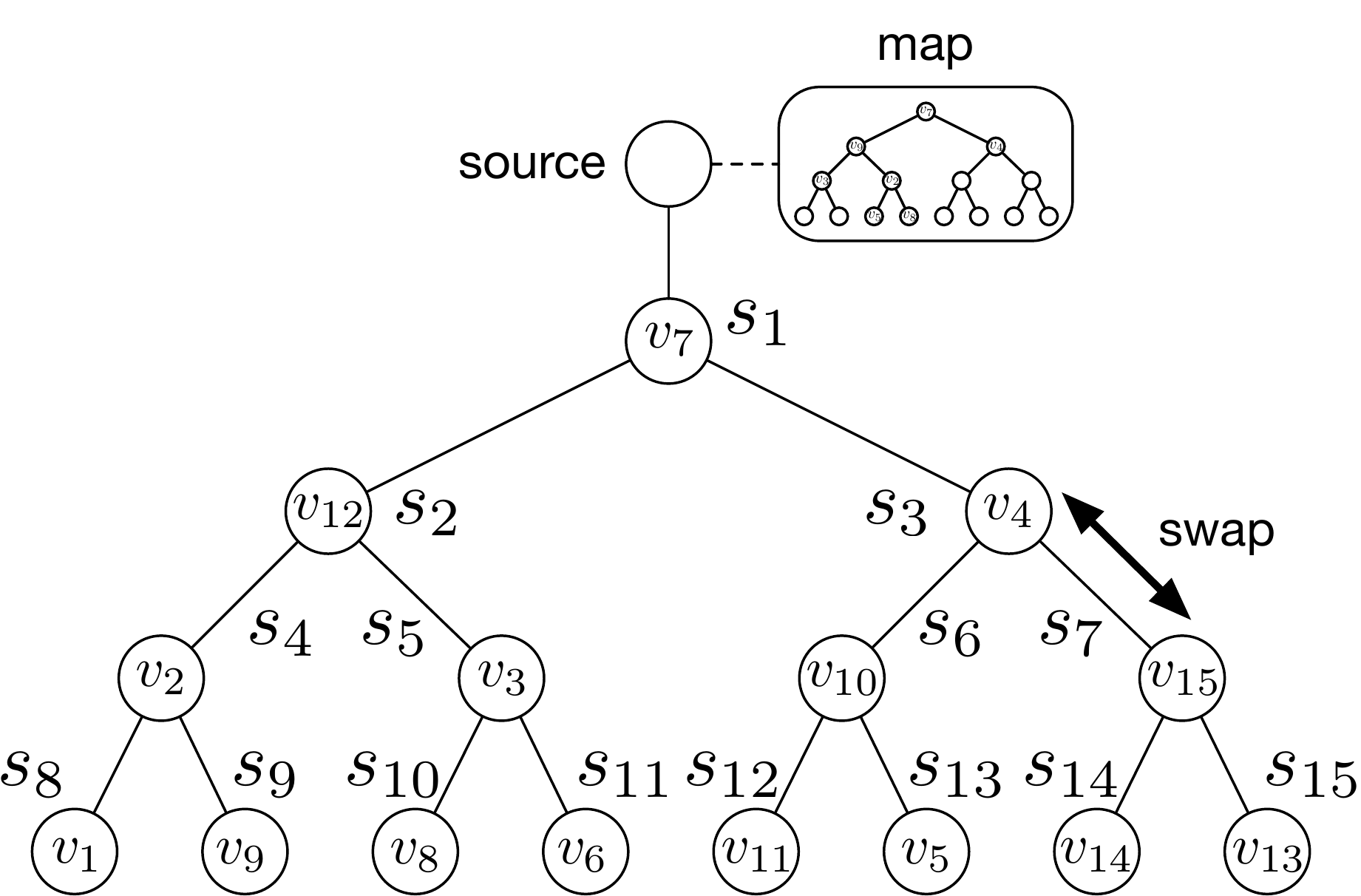} &
   \includegraphics[width=.32\textwidth]{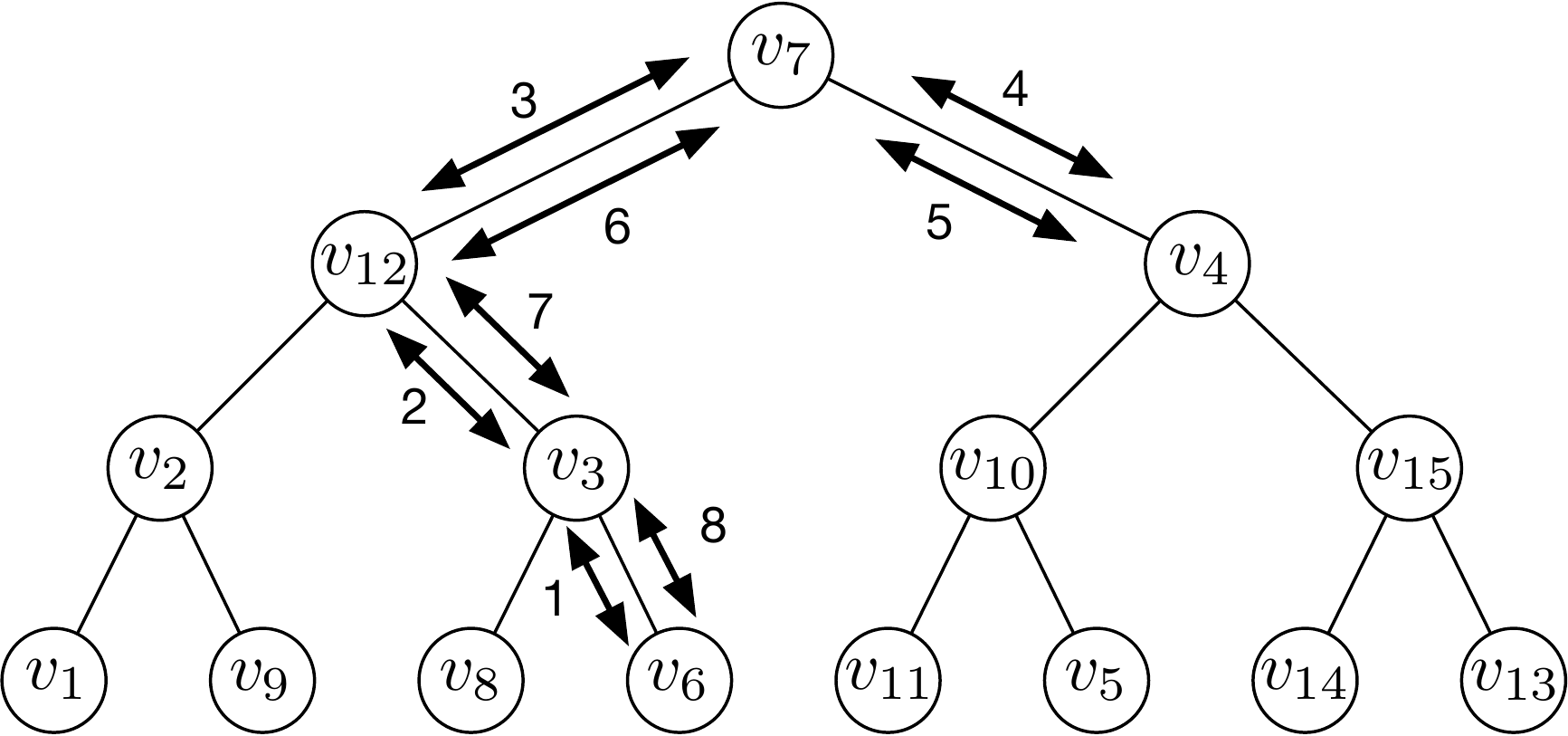} &
   \includegraphics[width=.32\textwidth]{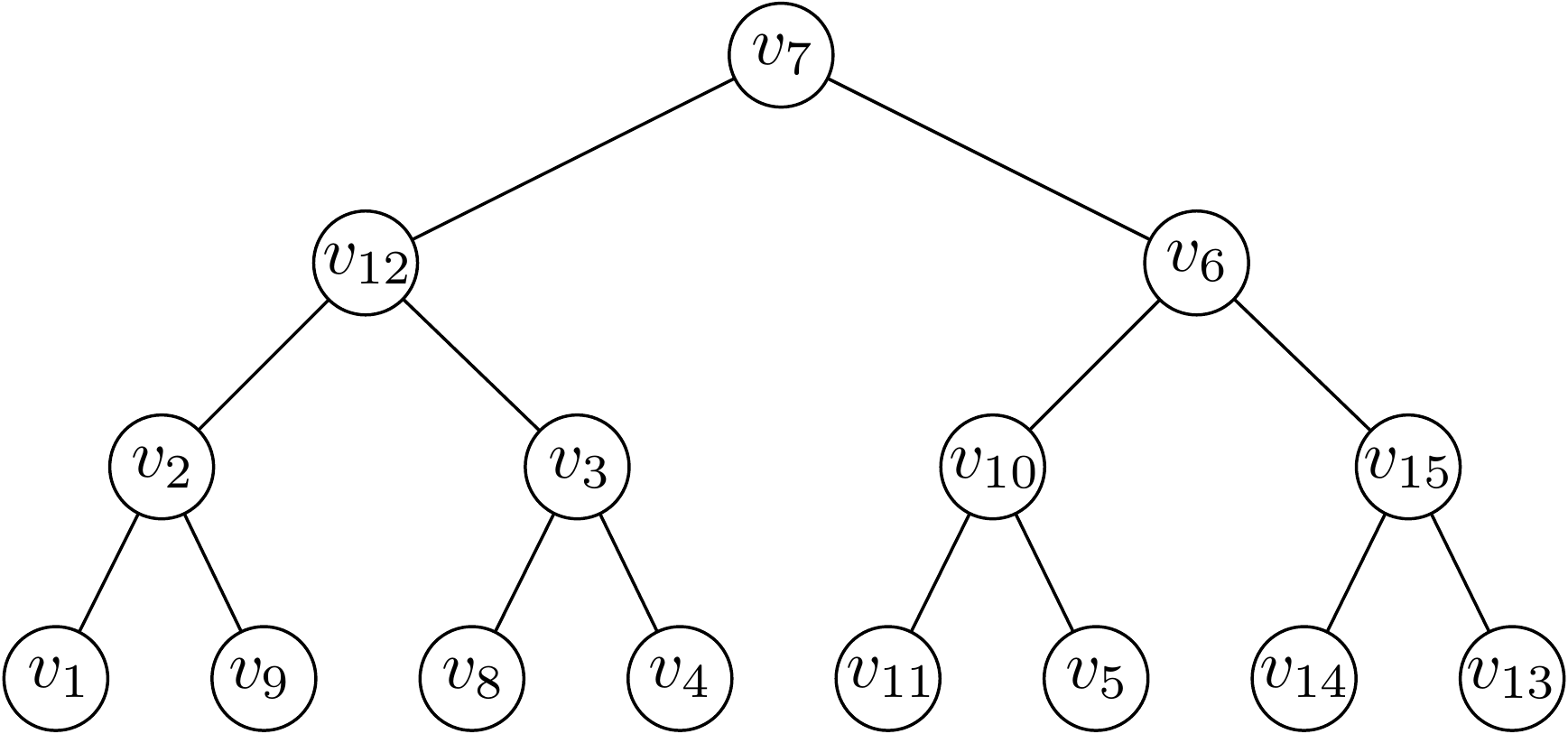} \\
    (a) & (b) & (c)
\end{tabular}
    \caption{(a) Our \emph{complete tree} model: a source with a map, a tree of servers that host items (nodes) and a \emph{swap} operation between neigboring items.
    (b)~The node's tree network implied by the tree $T$ from (a) and the set of swaps needed to interchange the location of $v_6$ and $v_4$. (c) The tree network after the interchange and swap operations of (b).}
    \label{fig:model}
 \end{centering}
\end{figure*}

\noindent{\bf Paper Organization:}We discuss problem model and other preliminaries in Section \ref{sec:prelim} followed by lower bound in Section \ref{sec:accessopt}. The deterministic algorithm $\ONDET$ with analysis is provided in Section \ref{sec:deterministic}. The randomized algorithm \textsc{\ONRAND} is discussed in Section \ref{sec:maintain}. Related work is in Section \ref{sec:rel}. We conclude in Section \ref{sec:conclusion} followed by an appendix.

\section{Model and Preliminaries}\label{sec:prelim}

Our problem can be formalized using the following simple model.
We consider a single \emph{source} that needs to communicate
with a set of $n$ nodes $V=\{v_1,\ldots,v_n\}$. The nodes are arranged
in a complete binary tree  and the source is connected the the root  of the tree.
While the tree describes a reconfigurable
\emph{network}, we will use terminology from
datastructures, to highlight this relationship
and avoid the need to introduce new terms.

We consider a complete tree $T$ connecting $n$  servers $S=\{s_1,\ldots,s_n\}$.
We will denote by $s_1(T)$ the root of the tree $T$, or
$s_1$ when $T$ is clear from the context, and
by $s_i.\Left$ (resp.~$s_i.\Right$)
the left (resp.~right) child of server $s_i$.
We assume that the $n$ servers
store $n$ items  (nodes) $V=\{v_1,\ldots,v_n\}$,
one item per server.
For any $i\in[1,n]$ and any time $t$, we will denote by $s_i.\guest^{(t)} \in V$
the item mapped to $s_i$ at time $t$. Similarly, $v_i.\host^{(t)} \in S$ denotes the server hosting
item $v_i$. Note that if $v_i.\host^{(t)}=s_j$ then $s_j.\guest^{(t)}=v_i$.

The \emph{depth} of
a server $s_i$ is fixed and describes the distance from the root;
it is denoted
by $s_i.\depth$, and $s_1.\depth = 0$. The depth of
an item $v_i$ \emph{at time} $t$ is denoted by $v_i.\depth^{(t)}$,
and is given by the depth of the server to which $v_i$ is mapped at time $t$.
Note that $v_i.\depth^{(t)}=v_i.\host.\depth^{(t)}$.

To this end, we interpret communication
requests from the source as \emph{accesses} to \emph{items} stored
in the (unordered) tree.
All access requests (resp.~communication requests) to items
(resp.~nodes)
originate from the root $s_1$. If an item (resp.~node)
is frequently requested, it can make sense to
move this item (node) closer to the root of $T$: this
is achieved by \emph{swapping} items which are neighboring
in the tree (resp.~ by performing local topological swaps).

Access requests occur over time,
forming a (finite or infinite) sequence
$\sigma=(\sigma^{(1)},\sigma^{(2)},\ldots)$,
where $\sigma^{(t)} = v_i \in V$
denotes that item $v_i$ is requested, and needs
to be accessed at time $t$.
The  sequence $\sigma$ (henceforth also called the \emph{workload}) is revealed one-by-one to an
online algorithm $\ON$.
The \emph{working set} of an item $v_i$ at time $t$ is the set of distinct elements accessed since the last access
of $v_i$ prior to time $t$, including $v_i$.
We define the \emph{rank} of item $v_i$ at time $t$ to be the size of the working set of $v_i$ at time $t$ and denote it as $v_i.\rank^{(t)}$.
When $t$ is clear of context, we simply write $v_i.\rank$.
The working set bound of sequence $\sigma$ of $m$ requests is defined as
$WS(\sigma) = \sum_{t=1}^{m} \log(\sigma^{(t)}.\rank)$.

Both serving (i.e., \emph{routing}) the request and adjusting the configuration
comes at a cost. We will discuss the two cost components in turn.
Upon a request, i.e., whenever the source wants to communicate to a partner, it routes
to it via the tree $T$.
To this end, a message passed between nodes
can include, for each node it passes, a bit indicating which child to forward the message next
(requires $O(\log n$) bits). Such a \emph{source routing} header can be built based
on a dynamic global \emph{map} of the tree that is maintained at the source node.
As mentioned, the source node is a direct neighbor of the root of the tree, aware of all requests, and therefore it can maintain the map.
The \emph{access cost} is hence given by the distance between the root
and the requested item, which is basically the depth of the item in the tree.

The \emph{reconfiguration cost} is due to the
adjustments that an algorithm performs on the tree. We define the unit cost of reconfiguration as a \emph{swap}: a
swap means changing position of an item with its parent.
Note that, any two items $u, v$ in the tree can be
\emph{interchanged} using a number of swaps equal
to twice the distance between them. This can be achieved by $u$ first swapping along the path to $v$ and then
$v$ swapping along the same path to initial location of $u$. This interchange operation results in the tree staying the same, but only $u$ and $v$ changing locations.
We assume that to interchange items, we first need to access one of them.
See Figure \ref{fig:model} for an example of our model and interchange operation.

\begin{definition}[Cost]\label{def:cost}
The cost incurred by an algorithm $\ALG$ to serve a request $\sigma^{(t)}= v_i$ is denoted by $\cost(\ALG(\sigma^{(t)}))$, short $\cost^{(t)}$.
It consists of two parts,
\emph{access} cost, denoted $\acc\cost^{(t)}$, and \emph{adjustment} cost, denoted $\adj\cost^{(t)}$.
We define access cost simply as $\acc\cost^{(t)} = v_i.\depth^{(t)}$ since $\ALG$ can maintain a global \emph{map} and access $v_i$ via the shortest path.
Adjustment cost, $\adj\cost^{(t)}$, is the
total number of swaps, where a single swap
means changing position of an item with its parent or a child.
The total cost,
incurred by $\ALG$ is then
\begin{align*}
\cost(\ALG(\sigma)) =\sum_t \cost(\ALG(\sigma^{(t)})) = \sum_t \cost^{(t)} =\sum_t  (\acc\cost^{(t)}+\adj\cost^{(t)})
\end{align*}
\end{definition}
Our main objective is to design online algorithms that perform almost
as well as optimal offline algorithms (which know $\sigma$
ahead of time), even in the worst-case.
In other words, we want to devise online algorithms
which minimize the competitive ratio:

\begin{definition}[Competitive Ratio $\rho$]\label{def:comp-ratio}
We consider the standard definition of (strict) competitive ratio $\rho$, i.e.,
$\rho = \max_{\sigma} \cost(\ON)/\cost(\OPT)$
where $\sigma$ is any input sequence and
where $\OPT$ denotes the optimal offline algorithm.
\end{definition}
If an online algorithm is constant competitive, independently
of the problem input, it is called
\emph{dynamically optimal}.

\begin{definition}[Dynamic Optimality]\label{def:dyn-opt}
An (online) algorithm $\ON$ achieves \emph{dynamic optimality}
if it asymptotically matches
the offline optimum on every access sequence. In other words,
the algorithm $\ON$ is $O(1)$-competitive.
\end{definition}

We also consider a weaker form of competitivity
(similarly to the notion of \emph{search-optimality} in
related work~\cite{search-optimality}), and say that $\ON$ is
\emph{access-competitive} if we consider only the access cost of  $\ON$
(and ignore any adjustment cost) when comparing it to $\OPT$ (which needs to pay
both for access and adjustment).
For a randomized algorithm, we consider an oblivious online adversary
which does not know the random bits of the online algorithm a priori.

The
\noindent \textbf{\emph{Self-adjusting Complete Tree Problem}}
considered in this paper can then be formulated as follows:
Find an online algorithm which serves any
(finite or infinite) online request sequence
$\sigma$
with minimum cost (including both access and rearrangement costs), on a
self-adjusting complete binary tree.

\section{Access Optimality: A Working Set Lower Bound}\label{sec:accessopt}

For \emph{fixed} trees, it is easy to see that
keeping frequent items close to the root,
i.e., using a \emph{Most-Frequently Used} (MFU) policy,
is optimal (cf.~Appendix).
The design of online algorithms
for \emph{adjusting} trees is more involved.
In particular, it is known that MFU is not optimal for lists~\cite{list-update-upperbound}.
A natural strategy could be to try and keep items close to the root which have been frequent ``recently''. However, this raises the question over which time interval to compute the frequencies. Moreover, changing from one MFU tree to another one may entail high adjustment costs.

This section introduces a natural \emph{pendant} to the MFU tree for a dynamic setting: the \emph{Most Recently Used (MRU) tree}. Intuitively, the MRU tree tries to keep the ``working set'' resp.~\emph{recently} accessed items close to the root. In this section we show a working set lower bound for any self-adjusting complete binary tree.

While the move-to-front algorithm, known to be dynamically optimal for self-adjusting lists \cite{list-update-upperbound}, naturally provides such a ``most recently used'' property, generalizing move-to-front to the tree is non-trivial. We therefore first show that any algorithm that maintains an MRU tree is \emph{access-competitive}.
With this in mind, let us first formally define the MRU tree.

\begin{definition}[MRU Tree]\label{def:mru}
For a given time $t$, a tree $T$ is an \emph{MRU} tree
if and only if,
\begin{eqnarray}\label{eq:rankdepth}
v_i.\depth = \lfloor\log v_i.\rank \rfloor
\end{eqnarray}
\end{definition}

Accordingly the root of the tree (level zero) will always host an item of $\rank$ one. More generally, servers in level $i$ will host items that have a $\rank$ between $(2^{i},2^{i+1}-1)$. Upon a request of an item, say $v_j$ with $\rank$ $r$, the $\rank$ of $v_j$ is updated to one, and only the ranks of items with $\rank$ smaller than $r$ are increased, each by 1. Therefore, the $\rank$ of items with rank higher than $r$ do not change and their level (i.e., depth) in the MRU tree remains the same (but they may switch location within the same level).

\begin{definition}[MRU algorithm]
An online algorithm $\ON$ has the \emph{MRU} property (or the working set property) if for each time $t$, the tree $T^{(t)}$ that
$\ON$  maintains, is an \emph{MRU} tree.
\end{definition}

The working set lower bound will follow from the following theorem (Theorem \ref{thm:mrucomp}) which
states that any algorithm that has the \emph{MRU} property is \emph{access competitive}.
Recall that an analogous statement of Theorem \ref{thm:mrucomp}
is known to be true for a \emph{list}~\cite{list-update-upperbound}.
As such, one would hope to find a simple proof that holds
for complete trees, but it turns out that this is
not trivial, since $\OPT$ has more freedom in trees.
We therefore present a direct proof based on a potential function,
similar in spirit to the list case.

\begin{theorem}\label{thm:mrucomp}
Any online algorithm $\ON$ that has the \emph{MRU} property is 4 access-competitive.
\end{theorem}

\begin{proof}
Consider the two algorithms $\ON$ and $\OPT$.
We employ a potential function argument which is based on the difference in the items' locations between $\ON$'s tree and $\OPT$'s tree.
For any server $s_i$, we define a pair $(s_i,s_j)$ as \emph{bad} on a tree of some algorithm $A$
if $s_i.\depth < s_j.\depth$ but $s_i.\guest.\rank (A) > s_j.\guest.\rank (A)$,
i.e., $s_i$ is at a lower level although $s_j.\guest (A)$ has been accessed more recently. We observe that any bad pair $(s_i,s_j)$ for $s_i$ is an ordered pair, i.e., this pair is not bad for $s_j$. Also note that, for any server $s_i$, $s_i.\depth$ is same on any tree for any algorithm, what may differ is $s_i.\guest$ resp. $s_i.\guest.\rank$.
Since $\ON$ has the \emph{MRU} property it follows by definition that none of its pairs are bad.
Hence bad pairs appear only on $\OPT$'s tree.
Let, for any algorithm $A$, $\alpha_i(A)$ denote the number of \emph{bad} pairs for $s_i$ in $A$'s tree.
Let $B_i(A)$ be equal to one plus $\alpha_i(A)$ divided by the number of elements at level $s_i.\depth$.
More formally,
$B_i(A)=1+\frac{\alpha_i(A)}{2^{s_i.\depth}}$.
Define $B(A) = \prod_{i=1}^{n} B_i(A)$. Now we define the potential function $\Phi = \log B(\OPT)-\log B(\ON)$ which is based on the difference in the number of bad pairs between $\ON$'s tree and $\OPT$'s tree. According to our definition, $B(\ON)=1$ and hence $\Phi = \log B(\OPT)$. Therefore, from now onwards, we use $B$ resp.~$\log B$ instead of $B(\OPT)$ resp.~$\log B(\OPT)$. We consider the occurrence of events in the following order. Upon a request, $\ON$ adjusts its tree, then $\OPT$ performs the rearrangements it requires.

Let the potential at time $t$ be $\Phi$ (i.e., before $\ON$'s adjustment, after serving request $\sigma^{(t)}$, and before $\OPT$'s rearrangements between requests $t$ and $t+1$) and the potential after $\ON$ adjusted to its tree be $\Phi'$. Then the potential change due to $\ON$'s adjustment is
\begin{align*}
\Delta \Phi_1 = \Phi' - \Phi = \log B' - \log B = \log \frac{B'}{B}
\end{align*}
We assume that the initial potential is $0$ (i.e., no item was accessed). Since the potential is always positive by definition, we can use it to bound the \emph{amortized} cost of $\ON$, $\acost(\ON)$. Consider a request at time $t$ to an item at depth $k$ in the tree of $\ON$. The access-cost is $\cost^{(t)}(\ON)=k$ and we would like to have the following bound: $\acost^{(t)}(\ON) \le \cost^{(t)}(\ON) + \Delta \Phi$. Assume that the requested item $\sigma^{(t)}$ is at
server $s_{r}$ at depth $j$ in $\OPT$'s tree, so $\OPT$ must pay at least an access cost of $j$. Let $k$ be the depth of $\sigma^{(t)}$ in $\ON$ tree. First we assume that $j < k$.

Let us compute the potential after $\ON$ updated its MRU tree. For any server $s_i$ at depth lower than $j$ i.e., for which $s_i.\depth < j$,
it holds that $B'_i = B_i + 1/2^{s_i.\depth} \leq B_i + 1$: since the rank of the guest of the last accessed server, $s_r$, changed (to $1$) and hence $\alpha_i$ increases by 1 for all of them. That is, for all servers for which $s_i.\depth \ge j$ (excluding $s_{r}$), $B'_i = B_i$. The potential of the accessed server, $s_{r}$, will be $B'_{r}=1$, since its guest's rank becomes $1$.
Although due to the access, the rank of some other elements increase by 1, that does not affect the number of bad pairs.
Let the rank of the requested element $\sigma^{(t)}$ before it was accessed be $\sigma^{(t)}.\rank$. After the access at time $t$,
the rank of all the elements with rank lesser than $\sigma^{(t)}.\rank$ will increase by 1. Consider any pair $(s_i,s_j)$ before the access of $\sigma^{(t)}$. We have already seen what happens if either $s_i$ or $s_j$ is $s_r$.
Otherwise a pair $(s_i,s_j)$ cannot change from
bad to good (resp.~good to bad) since if only $s_j.\rank$ (resp.~$s_i.\rank$) increases by 1, it cannot be more than that of $s_i$ (resp.~$s_j$). Now:
\begin{align*}
B' = \prod_{i=1}^{n} B'_i &\le  \left( \prod_{s_i.\depth < j} (B_i+1) \right) \left (\prod_{\substack{s_i.\depth \ge j \\ i \neq r}} B_i \right ) B'_{r}\notag \\
&\le \left(2^{j} \prod_{s_i.\depth < j}B_i \right) B_r \left (\prod_{\substack{s_i.\depth \ge j \\ i \neq r}} B_i \right ) \frac{B'_{r}}{B_r} \notag\\&
= \frac{2^{j}}{B_r}\prod_{i=1}^{n} B_i =  \frac{2^{j}}{B_r} B
\end{align*}

The second line results from $\prod_{i=1}^{n}(B_i+1) \le 2^n\prod_{i=1}^{n}B_i$ when $B_i \ge 1$, and by multiplying and dividing by $B_{r}$. Also recall that $B'_{r}=1$.
Note that $s_r.\depth = j$ and $j < k$ so,
\begin{align*}
B_r \ge \frac{(2^{k+1}-1)- (2^{j+1}-1)}{2^j} \ge \frac{2^{k+1}}{2^j}-2\notag \ge\frac{2^{k+1}}{2^{j+1}}= 2^{k-j}
\end{align*}
Now consider the change in potential $\Delta \Phi_1$.
\begin{align*}
\Delta \Phi_1 = \log \frac{B'}{B} =
\log \frac{2^{j} B }{B_r B} =
\log \frac{2^{j}}{B_r} \notag
\le \log \frac{2^{j}}{2^{k-j}} \le
\log 2^{2j-k} = 2j-k
\end{align*}

Now we consider $j\geq k$. In this case also, for any server for which $s_i.\depth=j' < j$, it holds that $B'_i \leq B_i + 1$ and for all servers for which $s_i.\depth \ge j$ (excluding $s_{r}$), $B'_i = B_i$. Again $B'_{r}=1$ but $B_{r}\geq 1$ since $j\geq k$.
By similar calculations, we get $B'\frac{2^{j}}{B_r} B$ and then, $\Delta \Phi_1\le \log 2^j = j \le 2j-k$.
To complete the proof we need to compute the potential change due to $\OPT$'s rearrangements between accesses.
Consider the potential after $\OPT$ adjusted its tree, $\Phi''$. Then the potential change due to $\OPT$'s adjustment is
\begin{align*}
\Delta \Phi_2 = \Phi'' - \Phi'  = \log B'' - \log B' = \log \frac{B''}{B'}
\end{align*}

The only operation $\OPT$ performs is swap i.e., changing positions between parent and a child.
$\OPT$ may need to change positions of items during rearrangement between
accesses. These can always be done using multiple number of swaps, upwards or downwards or both.
Below we compute potential difference due to such a swap.
Let $\OPT$ access an item $z$ at $s_c$ from depth $k'$, raising it to depth $k'-1$ by swapping it with its parent $z_p$ at $s_{p}$.

For all servers with $s_i.\depth = k'-1$, except $s_{p}$, $B''_i \leq B'_i + 1/2^{k-1}$ holds, as $z_p$
goes to level $k'$ from $k'-1$ and may become bad to all the servers at level $k'-1$.
For $s_{p}$, $B''_{p} \leq 2B'_{c} + \frac{2^{k'}}{2^{k'-1}}=2B'_{c} +2\leq 4B'_{c}$, as all the items in layer $k'$ may become bad w.r.t.~$z$.
Also $B''_{c}\leq B'_{p}$
Notice that changes only occur at depth $k'-1$, nothing will change above or below that.
We use the following inequality while computing $B''$.

$$
\prod_{i=1}^n (x_i + 1/n) \le e \prod_{i=1}^n x_i~~~~~~\text{$x_i\geq 1 \forall i$}
$$
We prove it below.
\begin{align*}
\prod_{i=1}^n (x_i + 1/n)\le \prod_{i=1}^n (x_i + x_i/n)=\prod_{i=1}^n (x_i) (1+ 1/n)^n \le e \prod_{i=1}^n x_i
\end{align*}
The computation of $B''$ is shown in Table \ref{table:potential}.
\begin{table*}
\begin{align}
B'' = \prod_{i=1}^{n} B''_i &\le  \left( \prod_{s_i.\depth < k'-1} (B''_i) \right) \left (\prod_{\substack{s_i.\depth \ge k'\\ i \neq q}} B''_i \right ) (B''_c) (B''_{p})\left( \prod_{\substack{s_i.\depth = k'-1 \\ i \neq p}} (B''_i) \right)\notag \\
&\le  \left( \prod_{s_i.\depth < k'-1} (B'_i) \right) \left (\prod_{\substack{s_i.\depth \ge k'\\i\neq q}} B'_i \right ) (B'_{p})(4B'_{c}) \left( \prod_{\substack{s_i.\depth = k'-1\\ i \neq p}} (B'_i+ \frac{1}{2^{k'-1}}) \right)\notag \\
&= \left( \prod_{s_i.\depth < k'-1} (B'_i) \right) \left (\prod_{\substack{s_i.\depth \ge k'\\i\neq q}} B'_i \right ) (B'_{p})(4B'_{c})
\left( e\prod_{\substack{s_i.\depth = k'-1\\ i \neq p}} (B'_i) \right)\notag \\
&=4e\left( \prod_{s_i.\depth < k'-1} (B'_i) \right) \left (\prod_{\substack{s_i.\depth \ge k'}} B'_i \right ) (B'_{p})(B'_{c})
\left( \prod_{\substack{s_i.\depth = k'-1\\ i \neq p}} (B'_i) \right)\notag \\
&=4eB'\notag
\end{align}
\caption{\label{table:potential} Computation of $B''$}
\end{table*}

Now we compute the potential change due to $\OPT$'s single swap:
\begin{align*}
\Delta \Phi_2 = \log \frac{B''}{B'} &\leq\log {4e} < 4
\end{align*}

The potential change is less than 4 per swap where $\OPT$ must pay one for that swap. If the number of swaps is $m$ for the rearrangement of $\OPT$ between any two accesses, the potential change is bounded by $4m$.
Putting it all together, we get
\begin{align*}
\label{eq:ammortized}
\acost^{(t)}(\ON)& \le \cost^{(t)}(\ON) + \Delta \Phi_1 + \Delta \Phi_2 \notag
\leq k + (2j - k)+ 4m\\
& \leq 4(j+m)\notag  = 4\cost^{(t)}(\OPT)
\end{align*}

Finally,
\begin{align*}
 \cost(\ON) &= \sum^{t} \cost^{(t)}(\ON)\notag
 = \sum^{t} \acost^{(t)}(\ON) - (\Phi^{(t)} - \Phi^{(0)})\notag \le \sum^{t} \\
 & \acost^{(t)}(\ON) \le \sum^{t} 4\cost^{(t)}(\OPT) = 4\cost(\OPT)
\end{align*}
\end{proof}

Based on Theorem \ref{thm:mrucomp} we can now show our working set lower bound:
\wslowerbound*

\begin{proof}
The sum of the access costs of items from
an MRU tree is exactly $WS(\sigma)$.
For the sake of contradiction
assume that there is an algorithm $\ALG$
with cost $\cost(\ALG(\sigma)) < WS(\sigma)/4$.
If follows that Theorem \ref{thm:mrucomp} is not true. A contradiction.
\end{proof}

\section{Deterministic Algorithm}\label{sec:deterministic}

\subsection{Efficiently Maintaining an MRU Tree}

It follows from the previous section that if we can maintain an MRU tree at the cost of \emph{accessing} an MRU tree, we will have a dynamically optimal algorithm.
So we now turn our attention to the problem of efficiently maintaining an MRU tree.
To achieve optimality, we need that the tree adjustment cost will be
proportional to the access cost.
In particular, we aim to design a tree which on one hand achieves a good
approximation of the MRU property to capture temporal locality, by providing
fast \emph{access} (resp.~\emph{routing}) to items;
and on the other hand is also adjustable at low cost over time.

Let us now assume that a certain item $\sigma^{(t)}=u$ is accessed
at some time $t$.
In order to re-establish the (strict) MRU property, $u$ needs to be promoted
to the root.
This however raises the question of where to move the item currently
located at the root, let us call it~$v$.
A natural idea to make space for $u$ at the root while
preserving locality, is to \emph{push down}
items from the root, including item $v$.
However, note that simply pushing items down along the path between
$u$ and $v$ (as done in lists) will result in a poor performance in the tree.
To see this,
let us denote the sequence of items along the path from $u$ to $v$
by $P=(u,w_1,w_2,\ldots, w_{\ell},v)$, where $\ell=u.\depth$,
\emph{before} the adjustment.  Now assume that the access sequence
$\sigma$ is such that it repeatedly cycles through the sequence
$P$, in this order. The resulting cost per request is in the order
of $\Theta(\ell)$, i.e., could reach $\Theta(\log{n})$ for $\ell=\Theta(\log{n})$.
However, an algorithm which assigns (and then fixes) the items in $P$
to the top $\log{\ell}$ levels of the tree, 
will converge to a cost of only $\Theta(\log{\ell})\in O(\log\log{n})$
per request: an exponential improvement.

Another basic idea is to try and keep the MRU property at every step. Let
us call this strategy \textsc{Max-Push}.
Consider a request to item $u$ which is at depth $u.\depth=k$.
Initially $u$ is moved to the root. Then
the \textsc{Max-Push} strategy chooses for each depth
$i< u.\depth$, the \emph{least} recently accessed (and with maximum rank) item from level~$i$: formally,
$w_i = \arg\max_{v \in V: v.\depth=i}v.\rank$.
We then push $w_i$ to the host of $w_{i+1}$. It is not hard to see that
this strategy will actually maintain a perfect MRU tree.
However, items with the maximum $\rank$
in different levels, i.e., $w_i.\host$ and $w_{i+1}.\host$,
may not be in a parent-child relation. So to push $w_i$ to $w_{i+1}.\host$, we may
 need to travel all the way from $w_{i}.\host$ to the root and then from the root to $w_{i+1}.\host$, resulting in a cost proportional
to $i$ per level $i$.
This accumulates a rearrangement cost of $\sum_{i=1}^k i > k^2/2$ to push all the items with maximum $\rank$ at each layer, up to layer $k$.
This is not proportional to the original access cost $k$ of the requested item and therefore,
leads to a non-constant competitive ratio as high as $\Omega(\log n)$.

Later, in Section \ref{sec:maintain}, we will
present a randomized algorithm that maintains a
tree that approximates an MRU tree at a low cost.
But first, we will present
a simple deterministic algorithm that does not directly maintain
an MRU tree, but has cost that is proportional to the MRU cost and is hence
dynamically optimal.

\subsection{The $\ONDET$ Algorithm}
In this section we propose a simple deterministic algorithm, $\ONDET$, that is proven to be dynamically optimal. Interestingly $\ONDET$ does not maintain the MRU property but its cost is shown to be competitive to the \emph{access cost} on an MRU tree, and therefore, to the
working set lower bound.

$\ONDET$ is described in Algorithm \ref{alg:movehalfup}. Initially,
$\ONDET$ and $\OPT$ start from the same tree (which is assumed w.l.o.g.~to be an MRU tree).
Then, upon a request to an item $u$, $\ONDET$ first accesses $u$ and then interchanges its position with node $v$ that is the highest ranked item positioned at half of the depth of $u$ in the tree. After the interchange the tree remains the same, only $u$ and $v$ changed locations.
See Figure \ref{fig:model} (b) for an example of $\ONDET$ operation where $v_6$ at depth 3 is requested and is then interchanged with $v_4$ at depth 1 (assuming it is the highest rank node in level 1).

\begin{algorithm}[t]
\caption{Upon request to $u$ in \textsc{$\ONDET$'s Tree}}
\label{alg:movehalfup}
\begin{algorithmic}[1]
\vspace{2mm}
\STATE~\textbf{access} $u=s.\guest$ along the tree branches \hfill (cost: $u.\depth$)
\STATE~let $v$ be the item with the highest rank at depth $\lfloor u.\depth/2 \rfloor$
\STATE~\textbf{swap} $u$ along tree branches to node $v$ \hfill (cost: $\frac{3}{2}u.\depth$)
\STATE~\textbf{swap} $v$ along tree branches to server $s$ \hfill (cost: $\frac{3}{2}u.\depth$)
\end{algorithmic}\label{alg:movehalfup}
\end{algorithm}

The \emph{access cost} of $\ONDET$ is proportional to the access cost of an MRU tree.
\begin{theorem}\label{thm:detvs0}
Algorithm $\ONDET$ is 4 access-competitive to an MRU algorithm.
\end{theorem}

Before going to the proof of Theorem \ref{thm:detvs0}, we discuss several properties of $\ONDET$.
First, we show that whenever any element $v$ moves down in $\ONDET$'s tree,
its depth is  at most twice plus one when compared to its depth in an
MRU tree.
\begin{lemma}\label{lem:twicedepth}
Whenever some element $v$ moves down to depth $h$ in $\ONDET$'s tree, it is at least at depth $\lfloor h/2 \rfloor$ in an MRU tree.
\end{lemma}
\begin{proof}
Upon a request of some element $u$, say, from depth $h$, let $v$ replace
$u$ at depth $h$ in $\ONDET$'s tree, from depth $\lfloor h/2 \rfloor$.
At the time of this request, $v$ must be the highest ranked element at depth $\lfloor h/2 \rfloor$ and accordingly, is replaced by $u$.
As the depth of the root is zero, the total number of
elements in depth $\lfloor h/2 \rfloor$ is exactly $2^{\lfloor h/2 \rfloor}$.
So $v.\rank\geq 2^{\lfloor h/2 \rfloor}$ at the time $u$
is requested. Accordingly the position of $v$ in an MRU tree
is at least at depth  $\lfloor h/2 \rfloor$ (see Equation \ref{eq:rankdepth}).
Therefore, the depth of $v$ in $\ONDET$'s tree is at most twice plus
one when compared to its depth in an MRU tree.
\end{proof}
Next, let $t = 0$ or a time where an item $v$ was moved down in $\ONDET$'s tree. Let $t' > t$ be the first time that $v$ was requested in $\sigma$ after time $t$. Then we can claim the following:
\begin{claim}\label{claim:twice}
If the depth of $v$ in $\ONDET$'s tree is $h$ at time $t'$, then its depth in an MRU tree at time $t'$ is at least $\lfloor h/2 \rfloor$.
\end{claim}
\begin{proof}
For the case $t=0$, since initially $v$ is at the same depth in both
trees, the claim follows trivially.
If $t > 0$, then let $t''$ be the most recent time before $t'$ that $v$ was moved down. Then at time $t''$, item $v$ was moved from some depth $\lfloor h/2 \rfloor$ to $h$. At time $t''$, according to Lemma \ref{lem:twicedepth}, the depth of $v$ in an MRU tree was at least $\lfloor h/2 \rfloor$. Clearly $v$'s depth remains unchanged in $\ONDET$'s tree at time $t'$, since time $t''$ was the most recent move down of $v$. Also since we consider the first request of $v$ after time $t$, it means that the rank of $v$ could only increase between $t$ and $t'$.
So its depth in an MRU tree could not decrease from $\lfloor h/2 \rfloor$.
\end{proof}

We can now prove Theorem \ref{thm:detvs0}.

\begin{proof}[Proof of Theorem \ref{thm:detvs0}]
We analyze the access costs for an arbitrary item $u$ during the entire run of the algorithm.
Let $t_1$ be the time of the first request to $u$ during the execution of $\sigma$. Let $d_1$ be the first time that $u$ was moved down by $\ONDET$. Then define $t_i$, $i > 1$ to be the first time after time $d_{i-1}$ that $u$ is requested. And let
$d_i$ be the first time after $t_i$ that $u$ is moved down by $\ONDET$.
Assume that the depth of $u$ at time $t_i$ is $L$. Then according to Claim \ref{claim:twice} its depth
at an MRU tree is $\lfloor L/2 \rfloor$.
Let $t_i^1, t_i^2, \dots, t_i^j$ denote all the requests for $u$ between $t_i$ and $d_i$. A total of $j$ requests to $u$ without any move down of $u$ by $\ONDET$. We can bound the access cost of $\ONDET$ on these requests as follows. If $j=1$ it is $L$, if $j>1$ then:
\begin{align*}
access(u,t_i^1, t_i^j)(\ONDET) \le L+ \lfloor L/2 \rfloor+...+\lfloor L/2^{j-1} \rfloor \leq 2 L
\end{align*}
On the other hand the access cost of an MRU algorithm for the same set of requests is bounded as follows. If $j=1$ it is $\lfloor L/2 \rfloor$, if $j>1$ then,
\begin{align*}
access(u,t_i^1, t_i^j)(\MRU) \geq \lfloor L/2 \rfloor+ 1+1+...+1 \geq \lfloor L/2 \rfloor + j - 1 \ge L/2
\end{align*}
Therefore, for each $i$ we have
\begin{align*}
\frac{access(u,t_i^1, t_i^j)(\ONDET)}{access(u,t_i^1, t_i^j)(\MRU)}\leq 4
\end{align*}
This leads to the results that the total access cost for $u$
in $\ONDET$ is $4$-competitive to the total access for $u$
in an MRU tree.
Since this is true for each item in the sequence,
$\ONDET$ is 4-access competitive compared to an $MRU$ tree.
\end{proof}

\dynamicopt*
\begin{proof}[Proof of Theorem \ref{th:dynamicopt}]
Using Theorem \ref{thm:mrucomp} and Theorem~\ref{thm:detvs0}, $\ONDET$ is 16-access competitive.
It is easy to see from Algorithm \ref{alg:movehalfup} that total cost of  $\ONDET$'s tree is 4 times the access cost.
Considering these, $\ONDET$ is 64-competitive.
\end{proof}

In the coming section we show techniques to maintain MRU trees cheaply.
This is another way to maintain dynamic optimality.

\section{Randomized MRU Trees}\label{sec:maintain}

The question of how, and if at all possible, to maintain an MRU tree deterministically (where for each request $\sigma^{(t)}$, $\sigma^{(t)}.depth = \lfloor\log \sigma^{(t)}.\rank \rfloor$) at low cost is still an open problem.
But, in this section we show that the answer is affirmative with two relaxations:
namely by using randomization and approximation.
We believe that the properties of the algorithm we describe next
may also find applications
in other settings, and in particular data structures like skip lists~\cite{dean2007exploring}.

At the heart of our approach lies an algorithm to maintain
a constant approximation of the MRU tree at any time.
First we define MRU$(\beta)$ trees for any constant $\beta$.
\begin{definition}[MRU($\beta$) Tree]\label{def:blru}
A tree $T$ is called an \emph{MRU}$(\beta)$ tree
if it holds for any item $u$ and any time that, $u.\depth = \lfloor\log u.\rank \rfloor +\beta$.
\end{definition}
Note that, any MRU$(0)$ tree is also an MRU tree.
In particular, we prove in the following that a constant additive approximation
is sufficient to obtain dynamic optimality.

\begin{theorem}\label{thm:mrucompbeta}
Any online~MRU$(\beta)$ algorithm is $4(1+\lceil\frac{\beta}{2}\rceil)$ access-competitive.
\end{theorem}
\begin{proof}
According to Theorem \ref{thm:mrucomp}, MRU$(0)$ trees are 4 access-competitive. Here we only need to prove it for $\beta>0$.
Let us consider an algorithm $\ON$($\beta$) that maintains
an MRU($\beta$) tree for some $\beta$.
For each request $v$ that $\ON$($\beta$) needs to serve,
if $v.\depth = k$ is an MRU tree, then $\ON$($\beta$)  needs to pay,
in the worst case, $k+\beta$ (while $\ON$($0$) will pay $k$).
According to $\ON$($\beta$), the
item of rank 1 is always at depth 0 and
the item of rank 2 is always at level 1.
For every level $k \ge 2$, we have, $k+\beta \leq (1+ \lceil\frac{\beta}{2}\rceil) k$. For the special case of $k=1$, the item with rank $3$
can also be at most at depth $2$, so the formula holds. Overall, using Theorem \ref{thm:mrucomp} we have:
\begin{align*}
\cost(\ON(\beta)) \le (1+ \lceil\frac{\beta}{2}\rceil) \cost(\ON(0)) \le 4(1+ \lceil\frac{\beta}{2}\rceil) \cost(\OPT)
\end{align*}
Hence $\ON$ is  $4(1+ \lceil\frac{\beta}{2}\rceil)$ access-competitive.
\end{proof}

To efficiently achieve an MRU$(\beta)$ tree,
we propose the  \textsc{\ONRAND}
strategy (see Algorithm~\ref{alg:move2root}).
This is a simple randomized strategy which selects
a random path starting at the root, and
then steps down the tree to depth $k=u.\depth$
(the accessed item depth), by choosing uniformly at random between
the two children of each server at each step.
This can be seen as a
simple $k$-step random walk in a directed version of the tree,
starting from the root of the tree.
Clearly, the adjustment cost of \textsc{\ONRAND} is also proportional to $k$ and
its actions are independent of any oblivious online adversary. The main technical challenge of this section is proving the following theorem.
\begin{theorem}\label{th:randompush}
\textsc{\ONRAND} maintains an MRU$(4)$ (Definition \ref{def:blru}) tree in expectation, i.e., the expected depth of the item with rank $r$ is less than
$\log r + 3 < \lfloor \log r \rfloor + 4$ for any sequence $\sigma$ and any time $t$.
\end{theorem}

To analyze  \textsc{\ONRAND} and eventually prove Theorem \ref{th:randompush}, we will define several random variables for an arbitrary $\sigma$ and time $t$ (so we ignore them in the notation).
W.l.o.g., let $v$ be the item with rank $i$ at time $t$ and let $D(i)$ denote the depth of $v$ at time $t$.
First we note that by induction, it can be shown that the support of $D(i)$ is the set of integers $\{0,1, \dots, i-1\}$.
\begin{figure*}[t]
\begin{center}
\includegraphics[width=0.7\textwidth]{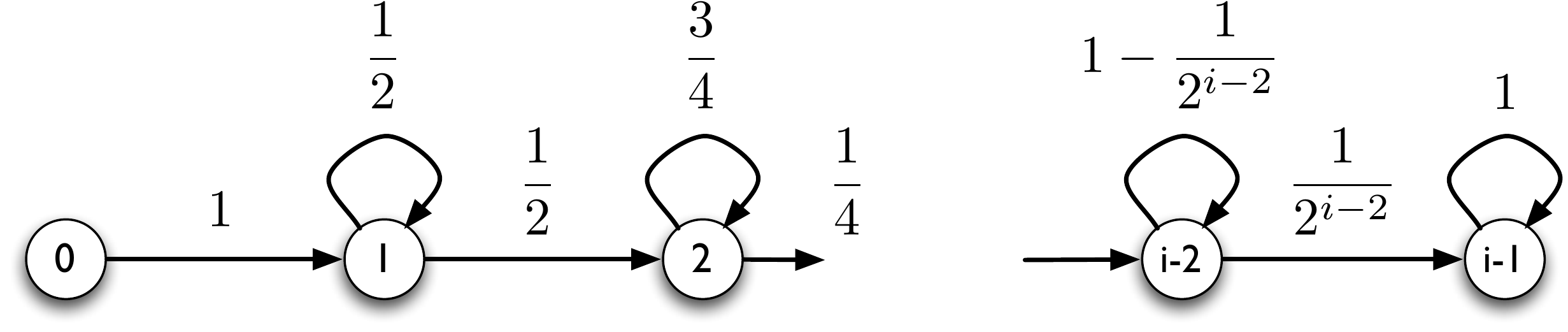}
\caption{The Markov chain $\mathcal{M}_i$ that is used to prove
Theorem~\ref{th:randompush} and Lemma \ref{lem:expst}:
possible depths for item of rank $i$ in the complete tree.}
\label{fig:chain}
\end{center}
\end{figure*}

To understand and upper bound $D(i)$,
we will use a Markov chain $\mathcal{M}_i$ over the integers
$0, 1, 2 \dots, i-1$, which denote the
possible depths in the \textsc{\ONRAND} tree, 
see Figure~\ref{fig:chain}.
For each depth in the chain $0 \le j < i-1$, the probability to move to depth $j+1$ is $2^{-j}$,
and the probability to stay at $j$,
is $1-2^{-j}$, for $j = i-1$; it is an absorbing state.
This chain captures the idea that the probability of an item at level $j$ to be pushed down the tree by a random walk
(to level larger than $j$) is  $2^{-j}$.
The chain \emph{does not} describe exactly our \textsc{\ONRAND} algorithm and $D(i)$, but we will use it to prove an upper bound on $D(i)$.
First, we consider a random walk described exactly by the Markov chain
$\mathcal{M}_i$ with an initial state $0$.
Let $\ell(i, w)$ denote the random variable of the last state of a
random walk of length $w$ on $\mathcal{M}_i$.
Then we can show:

\begin{lemma}
\label{lem:expst}
The expected state of $\ell(i, w)$ is such that ~$\exp[\ell(i, w)] <\lceil \log w \rceil + 1$,
and~$\exp[\ell(i, w)]$ is concave in $w$.
\end{lemma}

To prove Lemma \ref{lem:expst} we use the following corollary.
\begin{corollary}
\label{cor:series}
\begin{align*}
\sum_{i=0}^w {w \choose i} (\frac{w-1}{w})^{w-i} \cdot (\frac{1}{w})^i \cdot i = 1
\end{align*}
\end{corollary}
\begin{proof}
\begin{align*}
\sum_{i=0}^w & {w \choose i}(\frac{w-1}{w})^{w-i} (\frac{1}{w})^i i =\sum_{i=1}^w {w \choose i} (\frac{w-1}{w})^{w-i} (\frac{1}{w})^i i \\
&=\sum_{i=1}^w \frac{w!}{i!\ (w-i)!} (\frac{w-1}{w})^{w-i} (\frac{1}{w})^i i
=\sum_{i=1}^w \frac{(w-1)!}{(i-1)!\ (w-i)!} (\frac{w-1}{w})^{w-i} (\frac{1}{w})^{i-1} \\
&=\sum_{\eta=0}^{w-1} \frac{(w-1)!}{\eta!\ (w-1-\eta)!} (\frac{w-1}{w})^{w-1-\eta} (\frac{1}{w})^{\eta}~~~\text{(put $i=\eta+1$)}\\
&= \left(\frac{1}{w}+\frac{w-1}{w}\right)^{w-1}=1
\end{align*}
\end{proof}

\begin{proof}[Proof of Lemma \ref{lem:expst}]
First note that $\exp[\ell(i, w)]$ is strictly monotonic in $w$ and can be shown to be concave using the decreasing \emph{rate} of increase: for $w_1 < w_2$,
$\exp[\ell(i, w_1)] - \exp[\ell(i, w_1-\Delta)] \ge \exp[\ell(i, w_2)] - \exp[\ell(i, w_2-\Delta)]$.
To bound $\exp[\ell(i, w)]$, we consider the state of
another random walk, $\ell'(i, w)$, that starts on state $k=\lceil \log w \rceil$ in a modified chain $\mathcal{M}'_i$.
The modified chain $\mathcal{M}'_i$ is identical to $\mathcal{M}_i$ up to state $k$, but for all states $j > k$, the probability to move to state $j+1$ is $\frac{1}{w} > 2^{-k}$ and the probability to stay at $j$ is $\frac{w-1}{w} < 1-2^{-k}$. So clearly the walk on $\mathcal{M}'_i$ makes faster progress than the walk on $\mathcal{M}_i$ from state $k$ onward. The expected progress of the walk on  $\mathcal{M}'_i$ which starts from state $k$, is now easier to bound and can shown to be: $\exp[\ell'(i, w)] \le \sum_{j=0}^w j {w \choose j} (\frac{w-1}{w})^{w-j} (\frac{1}{w})^j = 1$.
But since $\ell(i, w)$ starts at state $0$, we have $\exp[\ell(i, w)] < \exp[\ell'(i, w)] \le k+1 = \lceil \log w \rceil +1$.
\end{proof}

Next, in Lemma  \ref{lem:walk}, we bound the expected number of times that $v$
could potentially be pushed down by a random push, i.e.,
the number of requests to elements at a lower depth than $v$. Later we will use this as the length $w$ of the random walk
on  $\mathcal{M}_i$. But, to do so we first state the following lemma.

\begin{lemma}\label{lem:monotone}
 For every $\sigma, t$ and $i > j$, we have that in \textsc{\ONRAND}, $\exp[D(i)] > \exp[D(j)]$.
\end{lemma}
\begin{proof}
Let $u$ be an item with rank $j < i$. Hence, it was requested
more recently than $v$ (which has rank $i$). The inequality follows from
the fact that conditioning
that $v$ and $u$ first reached the \emph{same}
depth (after the last request of $u$)
then by symmetry their expected progress of depth will be the same from that point.
More formally, let $D_{uv}$ be a random variable that denotes
the depth when $u$'s depth equals the depth of $v$ for the \emph{first} time
(since the last request of $u$ where its depth is set to 0);
and $-1$ if this never happens. Then by the law of
total probability,
\begin{align*}
\exp[D(i)] = \exp_{D_{uv}}[\exp[D(i) \mid D_{uv}]] \notag = \sum_{k=-1}^{i-1} \Prb(D_{uv}=k)  \sum_{\ell=0}^{i-1} \ell \cdot \Prb(D(i)\notag=\ell \mid D_{uv}=k)
\end{align*}
(and similar for $D(j)$).
But since
the random walk (i.e., push) is independent of the servers' ranks,
we have for $k \ge 0$ that $\exp[D(i) \mid D_{uv}=k] \ge \exp[D(j) \mid D_{uv}=k]$.
But additionally there is the possibility that they will never be at
the same depth (after the last request of $u$)
and that $v$ will always have a higher depth,
so $\exp[D(i) \mid D_{uv}=-1] > \exp[D(j) \mid D_{uv}=-1]$.
The claim follows.
\end{proof}

Now, let $W_i$ be a random variable that denotes the number of requests
for items with higher depth than $v$, since $v$'s last request until time $t$.
The following lemma bounds the number of such requests.

\begin{lemma}\label{lem:walk}
The expected number of requests for items with
higher depth than $v$, since $v$ was last requested,
is bounded by $\exp[W_i] \le 2i-1$.
\end{lemma}

\begin{proof}
We can divide $W_i$ into two types of requests:
$W_i^>$ requests for items
with higher rank and depth than $v$ at the time of their request, 
and $W_i^<$ requests for items with lower rank but higher depth than $v$ at the time of their request.
Then $W_i = W_i^> + W_i^<$. Clearly $W_i^>  \le i$
since every such request increases the rank of $v$ and
this happens $i$ times (note that some of these requests
may have lower depth than $v$).
$W_i^<$ is harder to analyze.
How many requests for items are there
that have lower rank than $v$ at the time of the request, but are below $v$ in the tree (i.e., have higher depth than $v$)? Note that such requests do not increase $v$'s rank, but \emph{may} increase its depth.
Let $u$ be an item with rank $j < i$, hence $u$ was more recently
requested than $v$ (maybe several times).
Let $X_j$ denote the number of requests for $u$ (since $v$ was last requested)
in which it had a higher depth than $v$.
Then $W_i^<  = \sum_{j=1}^{i-1} X_j$. We now claim that
$\exp[X_j] \le 1$. Assume by contradiction that  $\exp[X_j] > 1$.
But then this implies that we can construct a sequence $\sigma'$ for which the expected depth of $u$ will be larger than
the expected depth of $v$, contradicting Lemma~\ref{lem:monotone}.
Putting it all together:
\begin{align*}
\exp[W_i] = \exp[W_i^> + W_i^<] \notag \le \exp[i] + \exp[\sum_{n=1}^{i-1} X_n] \le i + \sum_{n=1}^{i-1} \exp[X_n] \le 2i-1
\end{align*}
\end{proof}

We now have all we need to prove Theorem \ref{th:randompush}.
The proof follows by showing that $\exp[D(i)] \le \log i + 3$.

\begin{proof}[Proof of Theorem \ref{th:randompush}]
Let $D(i,w)$ be a random variable that denotes the depth of
$v$ conditioning that there are $w$ requests of items with higher
depth than $v$, since the last request for $v$.
Note that by
the total probability law, we have that
\begin{align*}
\exp[D(i)] = \exp_{W_i}[\exp[D(i,W_i)] = \sum_{w=1}^{\card{\sigma}} \Prb(W_i=w) \exp[D(i,W_i=w)]
\end{align*}
Next we claim that $D(i,w)$ is stochastically less~\cite{shaked2007stochastic} than $\ell(i, w)$, denoted by
$D(i,w) \preceq \ell(i, w)$.

This is true since
the transition probabilities (to increase the depth) in the Markov chain
$\mathcal{M}_i$ are at least as high as in the Markov chain that describes $D(i,w)$.
The probability that a random walk to depth higher than $v$'s depth visits $v$ (and pushes it down)
is exactly $2^{-j}$ where $j$ is the depth of $v$.
Since $D(i,w) \preceq \ell(i, w)$, it will then follow from Theorem~\ref{thm:stocDom} that $\exp[D(i,w)] \le \exp[\ell(i, w)]$.
Clearly we also have
$\exp_{W_i}\exp[D(i,W_i)] \le \exp_{W_i}\exp[\ell(i, W_i)]$.
Let $f_i(W_i)= \exp[\ell(i, W_i)]$ be a random variable which is a function of the random variable~$W_i$. Recall that $f_i()$ is concave, then by Jensen's inequality~\cite{cover2006elements} and Lemma~\ref{lem:walk} we get:
\begin{align*}
\exp[D(i)] &=  \exp_{W_i}[\exp[D(i,W_i)]] \le \exp_{W_i}[\exp[\ell(i, W_i)]] \notag = \exp_{W_i}[f_i(W_i)] \le f_i(\exp[W_i]) \\
& \le f_i(2i-1)= \exp[\ell(i, 2i-1)] \le \lceil \log 2i \rceil  + 1 \le \log i + 3 \notag
\end{align*}
\end{proof}

\begin{algorithm}[t]
\caption{Upon access to $u$ in \textsc{Push-Down Tree}}
\label{alg:move2root}
\begin{algorithmic}[1]
\vspace{2mm}
\STATE~\textbf{access} $s=u.\host$ along tree branches ~~~~~~~~~~~ (cost: $u.\depth$)
\STATE~let $v=s_1.\guest$ be the item at the current root
\STATE~\textbf{move} $u$ to the root server $s_1$, setting $s_1.\guest = u$ \hfill(cost: $u.\depth$)
\STATE~employ \textsc{\ONRAND} to \textbf{shift} down $v$ to depth $s.\depth$ \hfill (cost: $u.\depth$)
\STATE~let $w$ be the item at the end of the push-down path, where $w.\depth=s.\depth$
\STATE~\textbf{move} $w$ to $s$, i.e., setting $s.\guest = w$  \hfill~~~~~~~~~~ (cost: $u.\depth \times 2$)
\end{algorithmic}\label{alg:move2root}
\end{algorithm}

It now follows almost directly from Theorems
\ref{thm:mrucompbeta}~and~\ref{th:randompush}
that \textsc{\ONRAND} is dynamically optimal.

\pushdowntree*
\begin{proof}
Let the $t$-th requested item
have rank $r_t$, then the access cost is $D(r_t)$.
According to the \textsc{\ONRAND} (Algorithm \ref{alg:move2root}), the total cost
is $5D(r_t)$ which is five times the access cost on the MRU$(4)$ tree.
Formally, using Theorem~ \ref{thm:mrucompbeta} and Theorem \ref{th:randompush}, the expected total cost is:
\begin{align*}
 \exp\left[\cost(\ONRAND)\right]& = \exp \left [\sum_{i=1}^{t} 5 D(r_i) \right] = 5\sum_{i=1}^{t} \exp \left [D(r_i) \right]\notag
  \le 5\sum_{i=1}^{t}  (\log (r_i) + 3) \notag \\& \le 5\sum_{i=1}^{t} (\lfloor \log (r_i) \rfloor+ 4)\notag
   \le  5 \sum_{i=1}^{t} \cost^{(t)}(\MRU(4)) \notag \\
   & \le  5\cdot \cost(\MRU(4))= 60 \cdot \cost(\OPT)
\end{align*}
\end{proof}

\section{Related Work}\label{sec:rel}

The self-adjusting tree networks considered in this paper
feature an interesting connection to self-adjusting
datastructures. A key difference is that while datastructures
need to be \emph{searchable}, networks come with \emph{routing}
protocols: the presence of a \emph{map} allows us to trivially
access a node (or item) at distance
$k$ from the front at a cost $k$.
Interestingly, while we have shown in this paper that
dynamically optimal algorithms for tree networks exist,
the quest for constant competitive online algorithms for
binary search trees remains a major open problem~\cite{splaytrees}. Nevertheless, there
are self-adjusting binary search trees that are known to be \emph{access optimal}~\cite{search-optimality}, but their rearrangement cost it too high.

In the following, we first review related work on datastructures
and then discuss literature in the context of networks.

\noindent \textbf{Dynamic List Update: Linked List (LL).}
The dynamically optimal linked list datastructure
is a seminal~\cite{list-update-upperbound}
result in the area: algorithms such as Move-To-Front (MTF),
which moves each accessed element to the front of the
list, are known to be 2-competitive,
which is optimal~\cite{albers1994competitive,albers1998self,list-update-upperbound}.
We note that the Move-To-Front algorithm results in the Most Recently Used property where items that were more recently used are closer to the head of the list.
The best known competitive ratio for randomized algorithms
for LLs is 1.6, which almost matches the randomized lower bound of 1.5~\cite{albers1995combined,teia1993lower}.

\noindent \textbf{Binary Search Tree (BST).}
In contrast to $\ct$s, self-adjustments in BSTs are based on \emph{rotations}
(which are assumed to have unit cost).
While BSTs have the working set property, we are missing
a matching lower bound: the \emph{Dynamic Optimality Conjecture},
the question whether splay trees~\cite{splaytrees} are dynamically
optimal,
continues to puzzle researchers even in the randomized case~\cite{albers2002randomized}.
On the positive side, over the last years,
many deep insights into  the properties of self-adjusting
BSTs have been obtained~\cite{tree-landscape},
including improved (but non-constant) competitive ratios~\cite{tango-trees},
regarding weaker properties such as
working sets,
static,
dynamic,
lazy,
and weighted,
fingers,
regarding pattern-avoidance~\cite{chalermsook2015pattern},
and so on.
It is also known
(under the name \emph{dynamic search-optimality}) that
if the online algorithm is allowed to make rotations for free after each
request, dynamic optimality can be achieved~\cite{search-optimality}.
Known lower bounds are by Wilber~\cite{wilber1989lower},
by Demaine et al.~\cite{almost}'s interleaves bound
(a variation), and by Derryberry et al.~\cite{derryberry2005lower}
(based on graphical interpretations).
It is not known today whether any of these lower bounds
is tight.

\noindent \textbf{Unordered Tree (UT).}
We are not the first to consider \emph{unordered} trees
and it is known that existing lower bounds for (offline) algorithms
on BSTs also apply to UTs that use rotations: Wilber's theorem can be generalized~\cite{fredman2012generalizing}.
However, it is also known that this correspondance between
ordered and unordered trees no longer holds under
weaker measures such as \emph{key independent processing costs} and in particular
\emph{Iacono's measure}~\cite{iacono2005key}:
the expected cost of the sequence
which results from
a random assignment of keys from the search tree to the items specified in an access
request sequence.
Iacono's work is also one example of prior work
which shows that for specific scenarios, working set and dynamic optimality
properties are
equivalent.
Regarding the current work, we note that the reconfiguration operations in UTs are more powerful than
the swapping operations
considered in our paper: a rotation allows to move entire subtrees at unit
costs, while the corresponding cost in $\ct$s is linear in the subtree size.
We also note that in our model, we cannot move freely between levels,
but moves can only occur between parent and child.
In contrast to UTs, $\ct$s are bound to be balanced.

\noindent \textbf{Skip List (SL) and B-Trees (BT).}
Intriguingly, although SLs and BSTs can be transformed
to each other~\cite{dean2007exploring}, Bose et al.~\cite{bose2008dynamic}
were able to prove dynamic optimality for (a restricted kind of) SLs as well as BTs.
Similarly to our paper, the authors rely on
a connection between dynamic optimality
and working set: they show that the working set property
is sufficient for their restricted SLs (for BSTs, it is known that
the working set is an upper bound, but it is not known yet
whether it is also a lower bound).
However, the quest for proving dynamic optimality for general
skip lists remains an open problem:
two restricted types of models were considered in~\cite{bose2008dynamic},
bounded and weakly bounded.
In the bounded model, the adversary can never forward more than $B$ times
on a given skip list level, without going down in the search;
and in the weakly bounded model, the first $i$ highest levels contain no more than $\sum_{j=0}^{i}  B^j$ elements. Optimality only holds for constant $B$. The weakly bounded model is related to a complete $B$-ary tree (similar to our complete binary tree), but there is no obvious or direct connection between our result and the weakly bounded optimality.
Due to the relationship between SLs and BSTs,
a dynamically optimal SL would imply a working set lower bound for BST.
Moreover, while both in their model and ours, proving the working set
property is key, the problems turn out to be fundamentally different.
In contrast to SLs,
$\ct$s revolve around \emph{unordered} (and balanced) trees
(that do not provide a simple search mechanism),
rely on a different reconfiguration operation (i.e., swapping or \emph{pushing}
an item to its parent comes at unit cost), and, as we show in this paper,
actually provide dynamic optimality for their general form. 
Finally, we note that~\cite{bose2008dynamic} (somewhat implicitly)
also showed that a random walk approach can achieve the working set
property; in our paper, we show that the working set property
can even be achieved deterministically and without maintaining MRU.

\noindent \textbf{Heaps and Paging.}
More generally, our work is also reminiscent of
\emph{online paging} models for \emph{hierarchies} of caches~\cite{yadgar2011management},
which aim to keep high-capacity nodes resp.~frequently accessed items
close to each other, however, without accounting for the reconfiguration
cost over time.
Similar to the discussion above, self-adjusting $\ct$s differ
from paging models in that in our model, items cannot move
arbitrarily and freely between levels (but
only between parent and child at unit cost).

\noindent \textbf{Self-Adjusting Networks.}
Finally, little is known about self-adjusting \emph{networks}.
While there exist several algorithms for the design
of \emph{static} demand-aware networks, e.g.~\cite{avin2017demand,infocom19dan,ancs18,singla2010proteus},
online algorithms which also
minimize reconfiguration costs are less explored.
The most closely related work to ours
are \emph{SplayNets}~\cite{infocom19splay,splaynet},
which are also based on a tree topology (but a searchable one).
However, \emph{SplayNets} do not provide any formal guarantees over time,
besides convergence properties in case of certain fixed demands.
\section{Conclusion}\label{sec:conclusion}

This paper presented a deterministic and a randomized online
algorithm for a fundamental building block of self-adjusting networked
systems based on reconfigurable topologies.
We believe that our paper opens several interesting
avenues for future research, e.g., related to the
design of fully decentralized and self-adjusting communication networks
based on more general topologies and serving more general communication
patterns.

\noindent \textbf{Acknowledgments.}
Research supported by the ERC Consolidator grant
\emph{AdjustNet} (agreement no.~864228).

\bibliographystyle{plain}
\bibliography{ref}

\begin{thebibliography}{10}

\bibitem{albers1994competitive}
Susanne Albers.
\newblock A competitive analysis of the list update problem with lookahead.
\newblock {\em Mathematical Foundations of Computer Science 1994}, pages
  199--210, 1994.

\bibitem{albers2002randomized}
Susanne Albers and Marek Karpinski.
\newblock Randomized splay trees: theoretical and experimental results.
\newblock {\em Information Processing Letters}, 81(4):213--221, 2002.

\bibitem{albers1995combined}
Susanne Albers, Bernhard Von~Stengel, and Ralph Werchner.
\newblock A combined bit and timestamp algorithm for the list update problem.
\newblock {\em Information Processing Letters}, 56(3):135--139, 1995.

\bibitem{albers1998self}
Susanne Albers and Jeffery Westbrook.
\newblock Self-organizing data structures.
\newblock In {\em Online algorithms}, pages 13--51. Springer, 1998.

\bibitem{sigmetrics20complexity}
Chen Avin, Manya Ghobadi, Chen Griner, and Stefan Schmid.
\newblock On the complexity of traffic traces and implications.
\newblock In {\em Proc. ACM SIGMETRICS}, 2020.

\bibitem{avin2017demand}
Chen Avin, Kaushik Mondal, and Stefan Schmid.
\newblock Demand-aware network designs of bounded degree.
\newblock {\em Distributed Computing}, 2017.

\bibitem{infocom19dan}
Chen Avin, Kaushik Mondal, and Stefan Schmid.
\newblock Demand-aware network design with minimal congestion and route
  lengths.
\newblock In {\em Proc. IEEE INFOCOM}, pages 1351--1359, 2019.

\bibitem{avin2019toward}
Chen Avin and Stefan Schmid.
\newblock Toward demand-aware networking: A theory for self-adjusting networks.
\newblock {\em ACM SIGCOMM Computer Communication Review}, 48(5):31--40, 2019.

\bibitem{search-optimality}
Avrim Blum, Shuchi Chawla, and Adam Kalai.
\newblock Static optimality and dynamic search-optimality in lists and trees.
\newblock In {\em Proc. 13th Annual ACM-SIAM Symposium on Discrete Algorithms
  (SODA)}, 2002.

\bibitem{eclipse}
Shaileshh Bojja~Venkatakrishnan, Mohammad Alizadeh, and Pramod Viswanath.
\newblock Costly circuits, submodular schedules and approximate
  carath{\'e}odory theorems.
\newblock In {\em Proceedings of the 2016 ACM SIGMETRICS International
  Conference on Measurement and Modeling of Computer Science}, pages 75--88,
  2016.

\bibitem{tango-trees}
Prosenjit Bose, Karim Dou{\"\i}eb, Vida Dujmovi{\'c}, and Rolf Fagerberg.
\newblock An o (log log n)-competitive binary search tree with optimal
  worst-case access times.
\newblock In {\em Scandinavian Workshop on Algorithm Theory}, pages 38--49.
  Springer, 2010.

\bibitem{bose2008dynamic}
Prosenjit Bose, Karim Dou{\"\i}eb, and Stefan Langerman.
\newblock Dynamic optimality for skip lists and b-trees.
\newblock In {\em Proc. 19th Annual ACM-SIAM Symposium on Discrete Algorithms
  (SODA)}, pages 1106--1114, 2008.

\bibitem{chalermsook2015pattern}
Parinya Chalermsook, Mayank Goswami, L{\'a}szl{\'o} Kozma, Kurt Mehlhorn, and
  Thatchaphol Saranurak.
\newblock Pattern-avoiding access in binary search trees.
\newblock In {\em Proc. Foundations of Computer Science (FOCS), 2015 IEEE 56th
  Annual Symposium on}, pages 410--423. IEEE, 2015.

\bibitem{tree-landscape}
Parinya Chalermsook, Mayank Goswami, L{\'a}szl{\'o} Kozma, Kurt Mehlhorn, and
  Thatchaphol Saranurak.
\newblock The landscape of bounds for binary search trees.
\newblock {\em arXiv preprint arXiv:1603.04892}, 2016.

\bibitem{cover2006elements}
T.M. Cover and J.~Thomas.
\newblock {\em Elements of information theory}.
\newblock Wiley, 2006.

\bibitem{dean2007exploring}
Brian~C. Dean and Zachary~H. Jones.
\newblock Exploring the duality between skip lists and binary search trees.
\newblock In {\em Proceedings of the 45th Annual Southeast Regional
  Conference}, ACM-SE 45, pages 395--399, New York, NY, USA, 2007. ACM.

\bibitem{almost}
Erik~D. Demaine, Dion Harmon, John Iacono, and Mihai Patrascu.
\newblock Dynamic optimality - almost.
\newblock {\em {SIAM} J. Comput.}, 37(1):240--251, 2007.

\bibitem{derryberry2005lower}
Jonathan Derryberry, Daniel~Dominic Sleator, and Chengwen~Chris Wang.
\newblock {\em A lower bound framework for binary search trees with rotations}.
\newblock School of Computer Science, Carnegie Mellon University, 2005.

\bibitem{ancs18}
Klaus-Tycho Foerster, Monia Ghobadi, and Stefan Schmid.
\newblock Characterizing the algorithmic complexity of reconfigurable data
  center architectures.
\newblock In {\em Proc. ACM/IEEE Symposium on Architectures for Networking and
  Communications Systems (ANCS)}, 2018.

\bibitem{sigact19}
Klaus-Tycho Foerster and Stefan Schmid.
\newblock Survey of reconfigurable data center networks: Enablers, algorithms,
  complexity.
\newblock In {\em SIGACT News}, 2019.

\bibitem{fredman2012generalizing}
Michael~L Fredman.
\newblock Generalizing a theorem of wilber on rotations in binary search trees
  to encompass unordered binary trees.
\newblock {\em Algorithmica}, 62(3-4):863--878, 2012.

\bibitem{firefly}
Navid Hamedazimi, Zafar Qazi, Himanshu Gupta, Vyas Sekar, Samir~R Das, Jon~P
  Longtin, Himanshu Shah, and Ashish Tanwer.
\newblock Firefly: A reconfigurable wireless data center fabric using
  free-space optics.
\newblock In {\em Proc. ACM SIGCOMM Computer Communication Review (CCR)},
  volume~44, pages 319--330, 2014.

\bibitem{iacono2005key}
John Iacono.
\newblock Key-independent optimality.
\newblock {\em Algorithmica}, 42(1):3--10, 2005.

\bibitem{kandula2009nature}
Srikanth Kandula, Sudipta Sengupta, Albert Greenberg, Parveen Patel, and Ronnie
  Chaiken.
\newblock The nature of data center traffic: measurements \& analysis.
\newblock In {\em Proc. 9th ACM Internet Measurement Conference (IMC)}, pages
  202--208, 2009.

\bibitem{ghobadi2016projector}
{M.~Ghobadi et al.}
\newblock Projector: Agile reconfigurable data center interconnect.
\newblock In {\em Proc. ACM SIGCOMM}, pages 216--229, 2016.

\bibitem{infocom19splay}
Bruna Peres, A~de~O Otavio, Olga Goussevskaia, Chen Avin, and Stefan Schmid.
\newblock Distributed self-adjusting tree networks.
\newblock In {\em Proc. IEEE INFOCOM}, pages 145--153, 2019.

\bibitem{splaynet}
Stefan Schmid, Chen Avin, Christian Scheideler, Michael Borokhovich, Bernhard
  Haeupler, and Zvi Lotker.
\newblock Splaynet: Towards locally self-adjusting networks.
\newblock {\em IEEE/ACM Trans. Netw.}, 24(3):1421--1433, June 2016.

\bibitem{shaked2007stochastic}
Moshe Shaked and J~George Shanthikumar.
\newblock {\em Stochastic orders}.
\newblock Springer Science \& Business Media, 2007.

\bibitem{singla2010proteus}
Ankit Singla, Atul Singh, Kishore Ramachandran, Lei Xu, and Yueping Zhang.
\newblock Proteus: a topology malleable data center network.
\newblock In {\em Proc. ACM Workshop on Hot Topics in Networks (HotNets)},
  2010.

\bibitem{list-update-upperbound}
Daniel~D. Sleator and Robert~E. Tarjan.
\newblock Amortized efficiency of list update and paging rules.
\newblock {\em Commun. ACM}, 28(2):202--208, February 1985.

\bibitem{splaytrees}
Daniel~Dominic Sleator and Robert~Endre Tarjan.
\newblock Self-adjusting binary search trees.
\newblock {\em J. ACM}, 32(3):652--686, July 1985.

\bibitem{teia1993lower}
Boris Teia.
\newblock A lower bound for randomized list update algorithms.
\newblock {\em Information Processing Letters}, 47(1):5--9, 1993.

\bibitem{wilber1989lower}
Robert Wilber.
\newblock Lower bounds for accessing binary search trees with rotations.
\newblock {\em SIAM Journal on Computing}, 18(1):56--67, 1989.

\bibitem{yadgar2011management}
Gala Yadgar, Michael Factor, Kai Li, and Assaf Schuster.
\newblock Management of multilevel, multiclient cache hierarchies with
  application hints.
\newblock {\em ACM Transactions on Computer Systems (TOCS)}, 29(2):5, 2011.

\end{thebibliography}

\section*{Appendix}

\section{Optimal Fixed Trees}\label{sec:static}

The key difference between binary search trees and
binary trees is that the latter provides more flexibilities in how
items can be arranged on the tree. Accordingly,
one may wonder whether more flexibilities will render the optimal
data structure design problem algorithmically simpler or harder.

In this section, we consider the static problem variant, and
investigate offline algorithms to compute optimal trees for a \emph{fixed}
frequency distribution over the items.
To this end, we assume
that for each item $v$, we are given a frequency $v.\freq$,
where $\sum_{v\in V} v.\freq = 1$.
\begin{definition}[Optimal Fixed Tree]\label{def:opt-stat}
We call a tree \emph{optimal static tree} if it minimizes the expected path length
$\sum_{i\in[1,n]} (v_i.\freq \cdot v_i.\depth)$.
\end{definition}

Our objective is to design an optimal static tree according to
Definition~\ref{def:opt-stat}.
Now, let us define the following notion of \emph{Most Frequently Used (MFU)} tree
which keeps items of larger empirical frequencies closer to the root:
\begin{definition}[MFU Tree]\label{def:mfu}
A tree in which for every pair of items $v_i,v_j \in V$,
it holds that if $v_i.\freq \geq v_j.\freq$ then
$v_i.\depth \leq v_j.\depth$,
is called \emph{MFU tree}.
\end{definition}

Observe that MFU trees are not unique but rather,
there are many MFU trees. In particular, the positions of
items at the same depth can be changed arbitrarily without violating the MFU properties.

\begin{theorem}[Optimal Fixed Trees]\label{thm:static-tree}
Any MFU tree is an optimal fixed tree.
\end{theorem}
\begin{proof}
Recall that by definition, MFU trees have the property that
for all node pairs $v_i,v_j$:
$v_i.\freq> v_j.\freq \Rightarrow v_i.\depth \leq v_j.\depth$.
For the sake of contradiction, assume that there is a tree $T$ which
achieves the minimum expected path length but for which there
exists at least one item pair $v_i,v_j$ which violates our assumption,
i.e., it holds that
$v_i.\freq>v_j.\freq$ but $v_i.\depth > v_j.\depth$.
From this, we can derive a contradiction to the minimum
expected path length: by swapping the positions of items
$v_i$ and $v_j$, we obtain a tree $T'$ with an expected path length
which is shorter by
$\cost(T, \sigma)-\cost(T', \sigma)=$
$v_i.\freq \cdot v_i.\depth+v_j.\freq \cdot v_j.\depth -
v_i.\freq\cdot v_j.\depth+v_j.\freq\cdot v_i.\depth$
$> 0$.
\end{proof}

MFU trees can also be constructed
very efficiently, e.g., by performing the following \emph{ordered insertion}:
we insert the items into the tree $T$ in a top-down,
left-to-right manner, in descending order of their frequencies (i.e.,
item $v_i$ is inserted before item $v_j$ if
$v_i.\freq>v_j.\freq$).

\section{Stochastic Domination}

We recall some of the known results related to Stochastic Domination~\cite{shaked2007stochastic}.
\begin{definition}[Stochastic Domination]
Let $\XX$ and $\YY$ be two random variables, not necessarily
on the same probability space. The random variable $\XX$ is
{\em stochastically smaller than} $\YY$, denoted by
$\XX\preceq\YY$, if $\Prb[\XX> z]\leq\Prb[\YY> z]$ for every $z\in\Rnum$.
If additionally $\Prb[\XX>z]<\Prb[\YY>z]$ for some $z$,
then $\XX$ is {\em stochastically strictly less than} $\YY$,
denoted by $\XX\prec \YY$.
\end{definition}
\begin{theorem}[Stochastic Order]
\label{thm:stocDom}
Let $\XX$ and $\YY$ be two random variables,
not necessarily on the same probability space.
\begin{enumerate}
	\item Suppose $\XX\prec \YY$.
	Then $\Expct[U(\XX)]<\Expct[U(\YY)]$ for any
	strictly increasing function $U$.
	\item Suppose $\XX_1\prec \YY_1$
	and $\XX_2\prec\YY_2$, for four random
	variables $\XX_1, \YY_1,\XX_2$ and $\YY_2$.
	Then $a\XX_1 + b\YY_1 \prec a\XX_2 + b\YY_2$
	for any two constants $a,b>0$.
	\item Suppose $U$ is a non-decreasing function and
	$\XX\prec \YY$ then $U(\XX)\prec U(\YY)$.
	\item Given that $\XX$ and $\YY$ follow the binomial distribution,
	i.e., $\XX\sim Bn(n_1,p_1)$ and $\YY\sim Bn(n_2,p_2)$,
	then $\XX\preceq \YY$ if and only if the following two
	conditions holds: $(1-p_1)^{n_1}\geq(1-p_2)^{n_2}$
	and $n_1\leq n_2$.
	\end{enumerate}
\end{theorem}

\end{document}